\documentclass{vldb_arxiv}

\usepackage{amsmath, amssymb, color, graphicx}

\newtheorem{thm}{Theorem}
\newtheorem{lemma}[thm]{Lemma}

\newtheorem{remark}[thm]{Remark}

\newtheorem{definition}[thm]{Definition}

\newtheorem{example}[thm]{Example}

\newcommand{\floor}[1]{\left\lfloor#1\right\rfloor}

\newcommand{\TV}{tree-node }

\newcommand{\TVs}{tree-nodes }

\newcommand{\problemname}{\textsc{Summarize}\ }

\DeclareMathOperator*{\argmax}{argmax}

\usepackage{color}

\begin{document}

\title{Hierarchical Summarization of Metric Changes}

\numberofauthors{3}
\author{
\alignauthor
Matthias Ruhl\\
       \affaddr{Accompani Inc}\\
       \affaddr{382 1st St}\\
       \affaddr{Los Altos, CA, USA}\\
       \email{matthias@accompani.com}
\alignauthor
Mukund Sundararajan\\
       \affaddr{Google Research}\\
       \affaddr{1600 Amphitheatre Pkway}\\
       \affaddr{Mountain View, CA, USA}\\
       \email{mukunds@google.com}
\alignauthor
Qiqi Yan\\
       \affaddr{Google Research}\\
       \affaddr{1600 Amphitheatre Pkway}\\
       \affaddr{Mountain View, CA, USA}\\
       \email{contact@qiqiyan.com}
}

\maketitle

\begin{abstract}
We study changes in metrics that are defined on a cartesian product of
trees. Such metrics occur naturally in many practical applications,
where a global metric (such as revenue) can be broken down along
several hierarchical dimensions (such as location, gender, etc).

Given a change in such a metric, our goal is to identify a small set
of non-overlapping data segments that account for the change. An
organization interested in improving the metric can then focus their
attention on these data segments.

Our key contribution is an algorithm that mimics the operation of a
hierarchical organization of analysts. The algorithm has been
successfully applied, for example within Google Adwords to help
advertisers triage the performance of their advertising campaigns.

We show that the algorithm is optimal for two dimensions, and has an
approximation ratio $\log^{d-2}(n+1)$ for $d \geq 3$
dimensions, where $n$ is the number of input data segments. For the Adwords application, we can show that our algorithm is in fact a $2$-approximation.

Mathematically, we identify a certain data pattern called a \emph{conflict} that both guides the design of the algorithm, and plays a central role in the hardness results.
We use these conflicts to both derive a lower bound of $1.144^{d-2}$ (again $d\geq3$) for our algorithm, and to show that the problem is NP-hard, justifying the focus on approximation.
\end{abstract}

\section{Motivation}

Organizations use metrics to track, analyze, and improve the
performance of their businesses. The organization might be a company,
a government, an advertiser or a website developer. And the metric
might be the revenue of the company, the level of unemployment in a
country, or the number of clicks for an online advertising campaign.
Indeed, our interest in this problem stems from creating tools that
help analysts reason about Google's revenue, and to help Google's
advertisers reason about the performance of their advertising
campaigns.

Metrics vary because of changes in the business environment. One
common task of a data scientist is to determine what drives changes in
a metric over time. In particular, they want to identify segments of
the business where the change is most pronounced. This helps
decision-makers within the organization to counter these changes if
they are negative, or amplify the changes if they are positive.

\begin{example}
\label{ex:running}
Consider a government data scientist analyzing an increase in
unemployment. She does this by comparing the current month's
employment data to the previous month's data to figure out what caused
the increase.

The domain of this data - the employment market - can be sliced along
many dimensions, such as geography, industry sector, demographics,
etc. into a very large collection of granular submarkets, each of
which has its own variation in employment. Naturally, this analysis
proceeds in two steps: 1) The summarization step: Identify a
\emph{small} set of submarkets which account for a \emph{majority} of
the variation in overall unemployment. 2) Design fixes for negative
trends. This second step is most often manually intensive and case
specific. Therefore, one hopes that the first step narrows focus
meaningfully.
\end{example}

It is commonly observed that hierarchical data lends itself naturally
to summarization (cf. OLAP~\cite{OLAP}). For instance, the geography
dimension in the above example has a natural hierarchy: metros,
states, countries, and so on. If all the metros in a state have
similar unemployment trends, it is more concise to report the state
slice as an output of the summary rather than each metro
separately. Industry sectors and demographics also have similar
hierarchical representations. Organizations are similarly structured
hierarchically. The hierarchies aid the allocation, coordination and
supervision of tasks that are intended to improve organizational
metrics (cf. Organizational Theory~\cite{orgtheory}). Just as
hierarchies in the data inform step 1 from Example~\ref{ex:running},
hierarchies in the organization help effective delegation required to
accomplish step 2. For instance, many companies have separate
hierarchies for Sales, Finance and Product functions. Data scientists
can help analyze performance metrics and assign responsibilities for
fixes to the right substructure of the right functional hierarchy.

When determining the submarkets that drive a metric change in step 1,
it is important to avoid ``double-counting'', which can happen if the
resulting submarkets are not independent. This is a problem, since it
is possible that the same drop in employment is visible in several
overlapping slices of data. For instance, a regional slice
(e.g. California), and an industry sector (e.g. construction), may
both account for the same underlying change in
unemployment. Double-counting causes several conceptual problems. It
might prevent responsibilities for a fix from being clearly assigned,
it might lead to duplicated efforts from two different parts of an
organization, or it might lead to the illusion of greater progress
than was actually made. For instance, if the construction division and
the California division both work on the same issue, or both take
credit for a fix, then this credit assignment does not add up. We will
therefore insist that the list of submarkets from step 1 are
\emph{non-overlapping}.

\textbf{Informal problem statement:} Identify a small list of non-overlapping sub-segments that
account for the majority of a variation in a key metric, where the
space of candidate sub-segments is determined by a few, hierarchical
dimensions.

As we discuss in Section~\ref{sec:related} in greater detail, there
are several formulations of the ``drill-down'' problem. Most of these
formulations attempt to summarize patterns in the data. They use a
combination of information-theoretic models and input from a human
analyst to do so. In contrast, we seek to summarize the source of a
high-level change in a metric. As we will show, this problem is more
amenable to nearly complete automation. Our model is inspired directly
by the excellent models of Fagin et al~\cite{Fagin} and
Sarawagi~\cite{Sarawagi}. Unfortunately, these papers used lattices to
model the hierarchical data. This leads to the combination of strong
hardness results and algorithms that work only for very a restricted
classes of hierarchies (essentially a single tree). The key
contribution of this paper is to model hierarchical data as a product
of trees, this leads us to identify an interesting dynamic-programming
algorithm called Cascading Analysts. This algorithm has good
worst-case performance guarantees, and it is natural---it mimics the
operation of a hierarchical organization of analysts. As we discuss in
Section~\ref{sec:application}, we have applied this algorithm within
large decision support products. One of these products helps
advertisers reason about their advertising spend, and the other helps
website developers reason about their traffic.

\section{Problem Statement}
\label{sec:problem}

With the motivation from the previous section, we are ready to state
our model and problem formally.

\subsection{Definitions}

We consider multi-dimensional datasets where each dimension can be
modeled as a rooted tree. In a single (rooted) tree $T$, we say that
two tree-nodes $p$ and $q$ {\em overlap} if either $p=q$ or they share
an ancestor-descendant relationship, otherwise they are called {\em
  non-overlapping}.

We extend this definition to a cartesian product of trees.  For a product of
trees $P =T_1 \times T_2 \times \dots \times T_d$, we say that two
nodes $p=(p_1,\dots,p_d)$ and $q=(q_1,\dots,q_d)$ overlap iff for
\emph{every} dimension $i$, the tree-nodes $p_i$ and $q_i$ overlap. Consequently,
if a pair of nodes does not overlap, then there exists a (possibly
non-unique) dimension $i$ such that the tree-nodes $p_i$ and $q_i$ do
not overlap, and we say that nodes $p$ and $q$ do not overlap
\emph{along} dimension $i$.

A node $p=(p_1,\dots,p_d)$ is in the subspace of a node
$q=(q_1,\dots,q_d)$ if for \emph{every} $i$, $p_i$ is either a descendant of
$q_i$ in the tree $T_i$, or $p_i = q_i$. We define $Sub(v)$ to be the set of nodes
that are in the subspace of $v$.  We denote the root of the product of the
trees as $r = (r_1, \ldots r_d)$, where each $r_i$ is the root of
$T_i$.

Finally, a set $S$ of nodes is \emph{overlap-free} if no two nodes in
$S$ overlap.

\begin{figure}
  \begin{center}
    \includegraphics[width=3in]{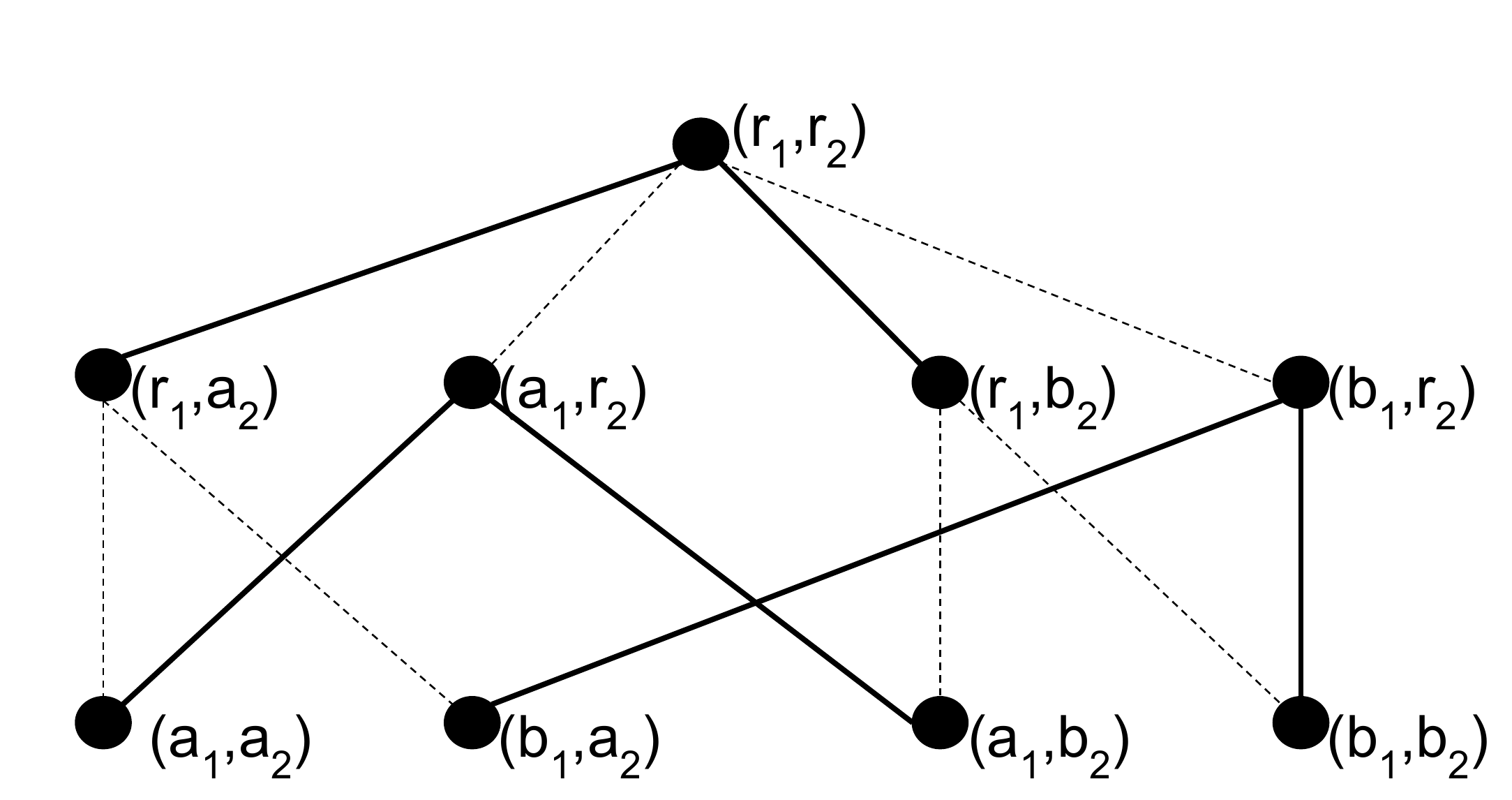}
  \end{center}
  \caption{The cartesian product of two trees $T_1$ and $T_2$, each of depth 2. Each node is a pair of tree nodes.
    The solid lines are edges along $T_2$ and the dotted lines are edges along $T_1$.}
  \label{fig:product}
\end{figure}

\begin{example}
Figure~\ref{fig:product} depicts the cartesian product of two trees $T_1$ and $T_2$, each of
depth $2$. Tree $T_i$ has root $r_i$ and left and right children $a_i$ and
$b_i$. Solid lines depict edges between parents and children along
tree $T_2$, and dotted lines depict the same relationship along tree
$T_1$. As examples of our definitions, nodes $(r_1,b_2)$ and $(r_1,a_2)$
do not overlap, but the pair $(r_1,b_2)$ and $(a_1,r_2)$ does.
\end{example}

\subsection{Formal Problem Statement}

With these definitions we can now formally state our problem.\\

\noindent {\sc Summarize}

\noindent {\bf Input:} Trees $T_1, T_2, \dots, T_d$ with the set of
tree-nodes $V := T_1 \times T_2 \times \dots \times T_d$, a non-negative
weight function $w : V \to \mathbb{R}^+$, and a maximum output size $k$.

\noindent {\bf Output:} Subset $S \subseteq V$, such that $|S| \leq
k$, $S$ is overlap-free, and $w(S) := \sum_{s \in S} w(s)$ is maximal
under these restrictions.\\

The main parameters of this problem are the number of dimensions $d$,
and $n := |V|$, the size of the input set. We express hardness and
runtime results in terms of these two variables.

\subsection{Modeling Metric Changes}

Let us briefly discuss how our problem definition relates to the
summarization of metric changes. Assume that we are comparing metrics for two points in time,
and that the value of the metric is defined for every leaf node in
Figure~\ref{fig:product}. Via aggregation, this defines the metric
values for every internal node as well.

Let us define the weight
function $w$ to be the absolute value of the difference in the metric
values for the two time points. (We discuss refinements to this weight function in Section~\ref{sec:weight-function}.)

Why does this weight function result in effective summarization? Consider two patterns of metric changes.
\begin{enumerate}
\item Two children of some node change their metrics in the same direction. Then the node's weight is the sum of children's weights.
\item Two children of some node change their metrics in different directions. Then the node's weight is smaller in magnitude than one of the children, possibly both. 
\end{enumerate}

In the first case, it is better to pick the node rather than the children. In the second case, it is better to pick one or both children rather than the parent. Further, notice that nodes may have multiple parents in different dimensions. So a given input may exhibit both patterns simultaneously, and then it is advantageous to summarize along certain dimensions.

To make this concrete, consider the topology in Figure~\ref{fig:product}.
Suppose that $k=2$ in the above
definition.  First let us suppose that $(a_1,a_2)$ and $(b_1,a_2)$ both go up by
some quantity $x$, while $(a_1,b_2)$ and $(b_1,b_2)$ fall by the same
quantity $x$. In essence, the change is along
dimension $2$.

Notice that the pair $(a_1,a_2)$ and $(b_1,a_2)$ and the pair $(a_1,b_2)$ and $(b_1,b_2)$ follow pattern 1 from above.
Whereas the pair $(a_1,a_2)$ and $(a_1,b_2)$ and the pair $(b_1,a_2)$ and  $(b_1,b_2)$ follow pattern 2.
Computing the weights shows us that the optimal
solution is the overlap-free set consisting of the nodes $(r_1, a_2)$
and $(r_1,b_2)$, reflecting the change along dimension $2$. Each of these nodes has a weight of $2x$. The other two
internal nodes and the root all have a weight of $0$.

\section{The Cascading Analysts Algorithm}
\label{sec:algorithm}

As we will discuss in Section~\ref{sec:negative}, {\sc Summarize}
cannot be solved optimally in polynomial time unless P=NP.  Therefore
we will now attempt to identify a good approximation algorithm for
{\sc Summarize}.  In this section, we describe the ``Cascading
Analysts'' algorithm that achieves that goal.

\subsection{Conflicts}
\label{sec:conflicts}

Our algorithm will achieve an optimal solution for a
more restricted version of {\sc Summarize}, namely where the solution
is additionally required to be \emph{conflict-free}.

The presence of a conflict prevents a set of nodes from being recursively subdivided one dimension at a time, even though the set of nodes is possibly overlap-free.
This definition and the example that follows elaborate.

\begin{definition}
A {\bf conflict} is a set $C \subseteq V$ such that for every
dimension $i$ there is a $(c_1,\dots,c_d) \in C$ such that for all $(x_1,\dots,x_d) \in C$, $x_i$
is a descendant of $c_i$, or $x_i=c_i$.
\end{definition}

A conflict can be overlap-free. Here is the simplest example of a conflict that is also overlap-free.
\begin{example}
\label{ex:simple}
  Consider three trees $T_1, T_2, T_3$, each of height two, and each with two leaves. Tree $T_i$ has root $r_i$ and left and right children $a_i$ and $b_i$ respectively.
  The conflict is defined by the set of nodes consisting of $(r_1,b_2,a_3)$, $(a_1, r_2, b_3)$ and $(b_1,a_2,r_3)$.
\end{example}

Conflicts play a central role in both the positive and the negative results in our paper. For instance, our algorithm will find the optimal \emph{conflict-free} solution, i.e., no subset of the nodes output by our algorithm contain a conflict. But it will only be approximate for {\sc Summarize}: given the input in Example~\ref{ex:simple}, our algorithm will only output two of the three nodes, even though the three nodes do not overlap. But we will show that the {\sc Summarize} problem is NP-hard even over problem instances of three trees of height two like the one in Example~\ref{ex:simple}, except with many more leaves. On the other hand, the optimal conflict-free solution can be found by recursively subdividing the product space of trees as we will show next. Each substep of our algorithm resembles a standard drill-down that an analyst would perform.

\begin{figure}
  \begin{center}
    \includegraphics[width=3in]{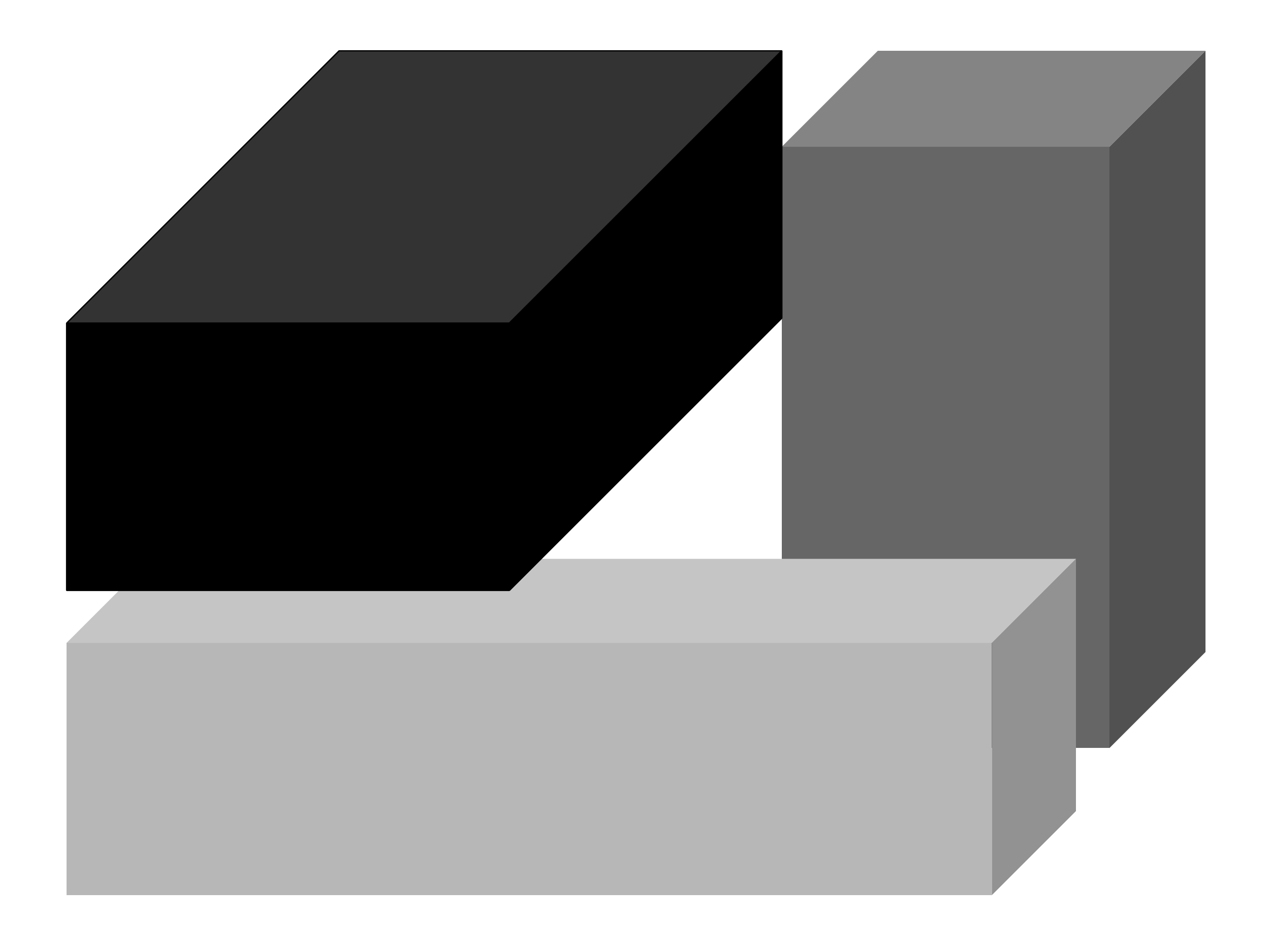}
  \end{center}
  \caption{Geometric representation of the conflict in Example~\ref{ex:simple}}
  \label{fig:conflict}
\end{figure}
The more visually inclined may benefit from a geometric view of a conflict. Figure~\ref{fig:conflict} depicts the conflict in Example~\ref{ex:simple}. The three trees in Example~\ref{ex:simple} correspond to the three dimensions in the figure. The three nodes correspond to the three cuboids in the figure. Note that for every dimension, there is a cuboid that spans it completely, i.e., every other cuboid ``overlaps'' with it in this dimension, or in the words of the the definition of a conflict, is a descendant of this cuboid in this dimension. This overlap prevents us from passing an axis aligned cutting plane through the space that would separate the cuboids without cutting at least one of them. This is the source of the approximation in our algorithm. Our algorithm will recursively perform axis aligned cuts and lose some of the nodes (cuboids) in the process. But as our hardness results in Section~\ref{sec:negative} show, these conflicts are exactly what make the problem inherently hard, so some of this loss is inevitable.

\subsection{Algorithm}
\label{sec:dp}

For each $v \in P$, our algorithm computes subsets $S(v,0)$, $S(v,1)$,
.., $S(v,k)$ such that for all $j \in \{ 0, \dots, k\}$: $S(v,j)
\subseteq Sub(v)$, $|S(v,j)| \leq j$, $S(v,j)$ is a conflict-free set,
and $S(v,j)$ is a maximum weight set satisfying these constraints.
The set $S(r,k)$ is the output of our algorithm, where we let $r := (r_1,
r_2, \dots, r_d)$ is the root of the product space. The sets $S(v,j)$
are computed bottom-up on the tree product space $P := T_1 \times T_2
\times \dots \times T_d$ via dynamic programming.\footnote{This dynamic program can either be implemented on a single machine, or distributed using MapReduce.}.\\

\textbf{1. Base case:} If $v$ is a leaf, i.e. $Sub(v) = \{ v \}$, then
we assign the sets in the ``obvious'' way: $S(v,0) := \emptyset$, and $S(v,j)
:= \{v\}$ for $j \geq 1$.\\

\textbf{2. Recursive step:} If $Sub(v) \neq \{ v \}$, we proceed as follows.
We let $S(v,0) := \emptyset$. For $j \in \{ 1, \ldots, k\}$, we set $S(v,j)$ as the maximum weight solution among these possibilities:
\begin{itemize}
\item The singleton set $\{ v \}$.
\item The $d$ or fewer solutions $S_i(v,j)$ that stem from repeating steps 2a and 2b below for those dimensions $i$ along which $v$ has children, i.e., $v_i$ is not a leaf of $T_i$.
 \end{itemize}

\textbf{2a. Breakdown along a dimension $i$:} Let $C_i(v)$ be the set of children of $v$ in dimension $i$. So if $v
= (v_1, v_2, \dots, v_d)$, then: $$C_i(v) = \{ (v_1, \dots, v_{i-1}, c,
v_{i+1}, \dots, v_d)\ |\ c\ \textrm{is child
  of}\ v_i\ \textrm{in}\ T_i \}$$

(This is the typical breakdown of the
space that an ``analyst'' interested in dimension $i$ would consider. The
algorithm performs a recursive sequence of these, hence the name ``Cascading
Analysts''.)

We let the maximal solution $S_i(v,j)$ along dimension $i$ be the largest weight union of sets of its children
$S(c_\ell, j_\ell)$ where $c_\ell \in C_i(v)$ (all $c_\ell$ distinct) and $\sum
j_\ell \leq j$. Note that the number of sets $S(c_\ell, j_\ell)$ can be anything
from $0$ to $j$. This step can be accomplished using a simple dynamic
program (not to be confused by the dynamic program over the tree structure) over the children. Here are the details:

\textbf{2b. Combining child solutions:}
This simple dynamic program orders the children in $C_i(v)$ in a sequence. Let $C_i^m(v)$ be the first $m$ nodes in $C_i(v)$ and $c^m$ be the $m$-th child in this sequence. Mirroring the definition of $S(v,j)$, let $S(C_i^m(v),j)$ be the optimal conflict-free solution of cardinality at most $j$ in the subspace of the first $m$ children of node $v$. The base case is when $m=1$; here we set $S(C_i^m(v),j)$ = $S(c^1,j)$. The recursive step is that we compute $S(C_i^m(v),j)$ by selecting the best union of $S(C_i^{m-1}(v),p)$ and $S(c^m,q)$, where $p+q \leq j$. Then, the optimal solution along dimension $i$ is defined by $S(C_i^\ell(v),j)$, where $\ell=|C_i(v)|$.

\begin{lemma}
\label{lem:conflict-free}
The Cascading Analysts algorithm will output a conflict-free set.
\end{lemma}

\begin{proof}
  We prove this claim bottom-up, mirroring the structure of the algorithm. For the base case, when $Sub(v) = \{ v \}$, $S(v,j)$ is either an empty set or a singleton set; both are conflict-free.
  When $Sub(v) \neq \{ v \}$, if $S(v,j) =\{v\}$ then it is conflict-free by itself. Otherwise, there is a dimension $i$ such that $S(v,j)$ is the union of sets each of which is contained within the subspace of a child $c \in C_i(v)$.

  Suppose that there is a conflict $Q$ in the union. Clearly $Q$ cannot be contained entirely with the subspace of any child, because inductively, these sets are each conflict-free. So $Q$ must span the subspace of at least two distinct children. But then $Q$ cannot be conflict-free because there is no node $q \in Q$ that is an ancestor along dimension $i$ to nodes in the subspace of both these children, violating the condition from the definition of conflicts. So we have a contradiction.
\end{proof}

\begin{lemma}
\label{lem:optimal-conflict-free}
The Cascading Analysts algorithm will output a maximum weight
overlap-free, conflict-free set.
\end{lemma}
\begin{proof}
We first show that the $S(v,j)$ are
overlap-free. This can be proved inductively. It is obviously true for
leaves of the product space (where $Sub(v) = \{v\}$). When combining
$S(c_\ell, j_\ell)$ from different children $c_\ell$ note that for elements of
$Sub(c_\ell)$ and $Sub(c_{\ell'})$ have empty intersection in the dimension
that we split on ($i$ in the above description). Thus, their union is
overlap-free.

Let $S_C$ be a maximum weight, overlap-free, conflict-free solution
for an instance of {\sc Summarize}. We will show that the weight of
the output of the Cascading Analysts algorithm is at least $w(S_C)$. Since the cascading analysts
algorithm outputs a conflict-free set (Lemma~\ref{lem:conflict-free}), this proves the lemma.

We show the following by induction: for all $v,j$, the weight of $S(v,j)$ is at least the
weight of the maximum weight conflict-free subset of size $j$ of
$Sub(v)$. Clearly this is true if $v$ is a leaf, so we just need
to consider the induction step.

Suppose $w(S(v,j)) < w(C)$, where $|C|=j$ and $C \subseteq Sub(v)$ is a
conflict-free set. Since $C$ is conflict-free, there has to be a dimension
$d'$ such that $\forall w \in C: w_{d'} \neq v_{d'}$. Let $D$ be the
children of $v$ in dimension $d'$. Then there are children $c_1,
\dots, c_\ell \in D$ so that $C = C_1 \cup \dots \cup C_\ell$ where $C_i
\subseteq Sub(c_i)$. We know that $w(S(c_i, |C_i|)) \geq w(C_i)$ by
induction hypothesis. The combination of these sets will be considered
by the algorithm, thus $w(S(v,j)) \geq \sum_i w(S(c_i, |C_i|)) \geq
w(C_i) = w(C)$, contradicting our assumption that the lemma did not
hold.
\end{proof}

\subsection{Running Time Analysis}
\label{sec:runtime}

There are $n$ choices of $v$ and $k$ choices of $j$ for which we need to compute the $S(v,j)$. With
standard techniques such as memoization, each of these needs to be touched
only once. For a fixed $v$, and fixed $j$, we can combine the child solutions to form $S(v,j)$ by a linear pass over the children (as in
step 2a of the algorithm in Section~\ref{sec:dp}). Each step in this pass incorporates an additional child and takes time $O(j)$.
Since a child can have at most $d$ parents, the total cost of this step for a single node is $O(dj)$.
Noting that $j \leq k$, this gives us a total runtime of $O(ndk^2)$.

\begin{remark}
Note that the size of the input grows multiplicatively in the number
of trees. For example, if each tree has size $10$, then $|V|$ is
$10^d$. Even reading the input becomes impractical for $d >
10$. Fortunately for us, there are compelling practical applications
where $d$ is fairly small, i.e. $d \leq 5$. This is true of our applications
in Section~\ref{sec:application}.
\end{remark}

\section{Performance Guarantee}

We first show that our algorithm is optimal for two dimensions,
followed by an approximation guarantee of $(\lceil
\log_{2}(n+1)\rceil)^{d-2}$ for the
case of three or more dimensions.

\begin{thm}
  \label{thm:twotree}
  The Cascading Analysts algorithm solves {\sc Summarize} optimally when $d=2$.
\end{thm}
\begin{proof}
  We show that when $d=2$, every overlap-free solution is also
  conflict-free. With Lemma~\ref{lem:optimal-conflict-free}, this
  concludes the proof.

  All we have to argue is that if a set of nodes $C$ constitutes a
  conflict, then it also contains an overlap. If $C$ is a conflict,
  then there are two nodes $x,y \in C$, $x= (x_1,x_2)$, $y
  =(y_1,y_2)$, such that for all $(c_1,c_2) \in C$, $x_1$ is an
  ancestor of $c_1$ and $y_2$ is an ancestor of $c_2$. If $x = y$ then
  this node overlaps with all other nodes in $C$, completing the
  proof. If $x \neq y$, then (as just stated) $x_1$ is an ancestor of
  $y_1$ and $y_2$ is an ancestor of $x_2$. Therefore $x$ and $y$
  overlap, completing the proof.
\end{proof}

\begin{thm}
\label{thm:log}
For $d \geq 2$, let $d$ trees $T_1,\dots,T_d$ have sizes $n_1, \dots,
n_d$ respectively. The Cascading Analysts algorithm is a $(\lceil
\log_{2}(m+1)\rceil)^{d-2}$-approximation algorithm for such an
instance of the {\sc Summarize} problem, where $m=\max_i n_i$.

Alternately, The Cascading Analysts algorithm is a $(\lceil
\log_{2}(n+1)\rceil)^{d-2}$-approximation algorithm.
\end{thm}

\begin{proof}
The second theorem statement is easily implied by the first because $m \leq n$, so we now prove the first statement.
Our proof is by induction over $d$. The base case for $d=2$ follows
from Theorem~\ref{thm:twotree}, so for the following assume $d>2$.

Given a tree, let an (ordered) {\em path} be a sequence of nodes
$v_1,\dots,v_h$ in the tree where $v_{i+1}$ is a child of $v_i$ for
each $i$.  We say that two paths $p_1, p_2$ {\em overlap} if some node $v_1\in
p_1$ overlaps with some node $v_2\in p_2$. The following combinatorial
lemma is fundamental to our proof.

\begin{lemma}
\label{lem:decomposition}
For every rooted tree with $\ell$ leaves, there exists a partition of
its nodes into $\lceil log_2(\ell+1)\rceil$ groups, such that each
group is a set of paths, and no two paths in a group overlap.
\end{lemma}

We defer the proof of Lemma~\ref{lem:decomposition}, and first use it
to finish the proof of Theorem~\ref{thm:log}. 

For a {\sc Summarize} instance $P$, let $Opt(P)$ denote its optimal
solution weight. Let $\beta := \lceil \log_{2}(m+1)\rceil$.  Using
Lemma~\ref{lem:decomposition}, we can decompose $T_1$ into disjoint
groups $T = G_1 \cup \dots \cup G_g$, where $g \leq \lceil \log_2 (n_1+1) \rceil \leq \beta$.

Let $P_{G_i}$ to denote the {\sc Summarize} problem restricted to $G_i
\times T_2 \times \cdots \times T_d$.  Then we have $\sum_i
Opt(P_{G_i}) \geq Opt(P)$. Wlog assume that $Opt(P_{G_1})$ has the
largest weight among the subproblems, and therefore $Opt(P_{G_1}) \geq
Opt(P)/g \geq Opt(P)/\beta$.

Recall that $G_1$ is a set of non-overlapping paths.  For each such
path $p$ in $G_1$, consider the {\sc Summarize} problem $P_p$ over $p
\times T_2 \times \cdots \times T_d$. Then we have $\sum_{p\in G_1}
Opt(P_p) \geq Opt(P_{G_1})$.

We remove the first dimension from $P_p$ to form a problem $P^{'}_p$
over $T_2 \times \cdots \times T_d$, by setting
$w'(t_2,\dots,t_d)$ to be $\max_{v\in p} w(v,t_2,\dots,t_d)$ for all
$(t_2,\dots,t_d)\in T_2\times\cdots\times T_d$.  Note that
$Opt(P'_p)=Opt(P_p)$, and a conflict-free solution for $P^{'}_{p}$ can
be mapped back to a conflict-free solution for $P_p$ with the same
weight, by replacing each $(t_2, \dots, t_d)$ by $(v, t_2, \dots,
t_d)$ where $v=\argmax_{v\in p} w(v,t_2,\dots,t_d)$.

Note that $P^{'}_p$ has $d-1$ dimensions. By inductive hypothesis, it
has a conflict-free solution $S'_p$ with $w'(S^{'}_p)\geq
Opt(P^{'}_p)/\beta^{d-3}$.  The corresponding conflict-free solution
$S_p$ for $P_p$ also satisfies $w(S_p)\geq Opt(P_p)/\beta^{d-3}$.
Since the $p$ in $G_1$ are non-overlapping, the union of solutions $S := \Cup_{p} S_p$ must be a conflict-free solution for $P_{G_1}$.

Combining our insights leads to $w(S) = \sum_p w(S_p) \geq \sum_p Opt(P_p)/\beta^{d-3} \geq Opt(P_{G_1})/\beta^{d-3} \geq Opt(P) / (\beta\beta^{d-3})= Opt(P) / \beta^{d-2}$.
\end{proof}


It remains to show Lemma~\ref{lem:decomposition}. We will prove the following slightly stronger
generalization to forests instead. Forests are sets of rooted trees,
and nodes in different trees in a forest are considered to be not
overlapping.

\begin{lemma}
\label{lem:decomposition-forest}
For every forest with $\ell$ leaves, there exists a partition of its
nodes into exactly $\lceil log_2(\ell+1)\rceil$ groups, such that each
group consists of a set of paths, and no two paths in a group overlap.
\end{lemma}

\begin{proof} We prove Lemma~\ref{lem:decomposition-forest} by induction on $\ell$.
Given a forest, let $v_1,\dots,v_\ell$ be a preorder traversal
ordering of its leaves.  Let $m :=\lceil \ell/2 \rceil$ be the index
of the middle leaf, and let $p$ be the path from $v_m$ all the way to
its root.  If $p$ contains all nodes in the forest, we are done.
Otherwise consider the forest over $v_1,\dots,v_{m-1}$, and the forest
over $v_{m+1},\dots,v_\ell$ respectively.  The two forests are
separated by the path $p$. In particular, no node in the first forest
overlaps with any node in the second forest.  We inductively apply the
lemma to the two (smaller) forests respectively, and obtain $\lceil
\log_2(m-1 + 1)\rceil$ groups for the first forest, and $\lceil
\log_2(\ell-m + 1)\rceil$ groups for the second forest, both of which
have size at most $\lceil log_2(\ell+1) \rceil - 1$.  No group from
the first forest overlaps with any group from the second forest.
Hence we can set $G_i$ to be the union of the $i$-th group for the
first forest with the $i$-th group for the second forest, for $i =
1,\dots,\lceil log_2(\ell+1) \rceil - 1$, and no two paths in $G_i$
overlap. Finally, we finish our construction by setting $G_{\lceil
  log_2(\ell+1) \rceil} := \{p\}$.

\end{proof}

\begin{remark}
  \label{rem:height}
  An alternative construction for Lemma~\ref{lem:decomposition-forest}
but with a looser bound is by an induction over the height of the
forest, where at each step, we take one root-to-leaf path from each
root of the forest, to form a group, and then proceed with the rest of
the forest. Each such step reduces the height of the forest by 1, and
we end up having as many groups as the height of the forest. In
practical applications, the depths of the hierarchies are usually
bounded by a small constant such as 3, so that this construction gives
a better approximation bound.
\end{remark}

\begin{remark}
  \label{rem:practice}
  In our applications, the worst-case approximation ratio is usually a small constant ($2-4$). These applications have at most two large dimensions, which do not contribute to the approximation ratio via the proof of Theorem~\ref{thm:log}; notice that the proof allows to leave out any two dimensions from the bound, and we may as well leave out the two dimensions that correspond to the largest trees. Further, the other dimensions all have height at most $2$. So by Remark~\ref{rem:height}, we get an approximation factor that is $2^{d-2}$, where $d$ is the number of dimensions. For instance, in the AdWords use-case, we get an approximation ratio of $2$. In practice, the approximation could be even better due to the absence of conflicts. We discuss this in Remark~\ref{rem:conflicts}.
\end{remark}

\section{Hardness and Lower Bounds}
\label{sec:negative}
\subsection{Hardness results}

We have seen that our algorithm solves {\sc Summarize} exactly for $d
\leq 2$, and provided approximation guarantees for $d \geq 3$. The
following theorem shows that an exact solution even for $d = 3$ is
likely infeasible.


\begin{figure}
  \begin{center}
    \includegraphics[width=3in]{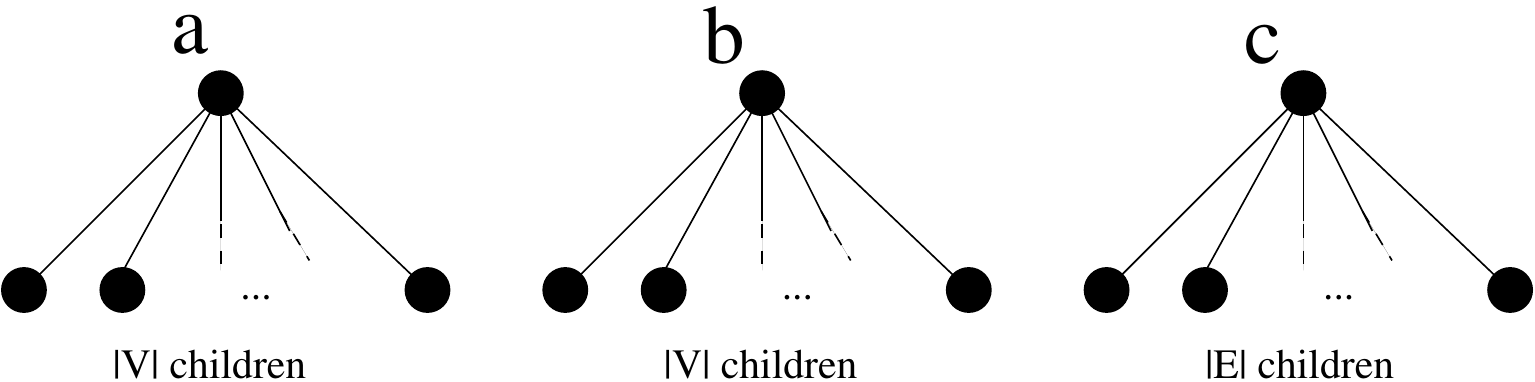}
  \end{center}
  \caption{A {\sc Summarize} problem instance with three trees of height two.}
  \label{fig:stars}
\end{figure}

\begin{thm}
\label{thm:ThreeTreeHardness}
{\sc Summarize} is NP-hard for $d = 3$.
\end{thm}

\begin{proof}
We show this by reduction from (directed) {\sc Maximum Independent Set
  (MIS)}. Given a directed graph $(V, E)$ (an instance of {\sc MIS}),
we construct an instance of {\sc Summarize} with $d=3$ as follows. Let
us call the three trees $A, B, C$ with roots $a,b,c$,
respectively. All the trees have height 2. In trees $A$ and $B$ we
have one child per vertex $v \in V$, called $a_v$ and $b_v$,
respectively. In the third tree, we have one child per edge $(v,w) \in
E$, called $c_{v,w}$ (see Figure~\ref{fig:stars}).

The weight function $w$ has non-zero weight on the following nodes:
\begin{itemize}
\item node $N_v := (a_v, b_v, c)$ has weight 1 for every $v \in V$, and
\item nodes $N_{v,w} := (a_v, b, c_{v,w})$ and $N'_{v,w} := (a, b_w, c_{v,w})$ have weight $\beta := 1 + \varepsilon$ for every $(v,w) \in E$.
\end{itemize}

We set $k = \infty$, so that any overlap-free set $S$ is a valid solution. Note that the reduction is polynomial-time; the number of nodes in the {\sc Summarize} instance is $O(|E|\cdot|V|^2)$.

We claim that there is a solution to {\sc MIS} of size $m$ if and only
if there is a solution $S$ to this instance of {\sc Summarize} with
$w(S) \geq m + \beta|E|$. This implies that {\sc Summarize} is NP-hard.

``$\Rightarrow$'': Let $V' \subseteq V$ be an independent set with
$|V'| = m$. Let $S$ be the union of the sets $S_1 := \{ N_v\ |\ v
\in V' \}$, $S_2 := \{ N_{v,w}\ |\ v \notin V' \}$, $S_3 := \{
N'_{v,w}\ |\ v \in V' \wedge w \notin V' \}$.

Since for all edges $(v,w)$ either $v$ or $w$ is not in $V'$, we have
that either $N_{v,w} \in S_2$ or $N'_{v,w} \in S_3$, and thus $|S_2
\cup S_3| = |E|$. Therefore, we have $w(S) = m + \beta|E|$. It remains to show that
$S$ is overlap-free. There are no overlaps in $S_2 \cup S_3$, since no
two elements overlap in the third dimension ($c_{v,w}$), neither are
there overlaps in $S_1 \cup S_2$, since elements differ in the first
dimension, and in $S_1 \cup S_3$, since elements differ in the second
dimension.

``$\Leftarrow$'': Given a solution $S$ to the {\sc Summarize} problem
with $w(S) \geq m + \beta|E|$, we need to construct an independent set
of size at least $m$. For each edge $(v,w)$, $S$ can contain at most
one of $N_{v,w}$ and $N'_{v,w}$, since the two nodes overlap. However,
wlog we can assume that $S$ contains exactly one of them: It is not
hard to see that if $S$ contains neither, and adding one of them would
create an overlap, then $S$ has to contain either $N_v$ or
$N_w$. However replacing e.g. $N_v$ by $N_{v,w}$ will increase the
weight of $S$ by $\beta - 1 = \varepsilon$. Note that this also implies that
for any edge $(v,w)$, $N_v$ and $N_w$ cannot both be in $S$.

So wlog, $S$ contains exactly one of $N_{v,w}$ and $N'_{v,w}$ for each
edge $(v,w)$. Since $w(S) \geq m + \beta|E|$, at least $m$ of the
weight comes from nodes
$N_v$. Thus, the set $V' := \{ v\ |\ N_v \in S \}$ satisfies $|V'|
\geq m$. As we observed before, for each edge $(v,w)$, not both $v$
and $w$ can be in $V'$, thus it is an independent set of the desired
size.
\end{proof}

\subsection{Lower bounds for the algorithm}

We now construct ``hard'' input instances for {\sc Summarize} for
which our algorithm outputs a solution that has $(2/3)^{d/3}$  of the weight of the optimal solution, when $d$ is a multiple of $3$. It follows that our
algorithm is at best a $(3/2)^{\floor{d/3}} > 1.144^{d-2}$ approximation
algorithm. Our strategy will be to construct an instance with lots of
conflicts.

\begin{thm}
  \label{thm:lb}
For every integer $m \geq 1$, there is an instance of {\sc Summarize}
with $d = 3m$ dimensions that has an overlap-free solution with weight
$3^m$, while the optimal conflict-free solution has weight $2^m$.
\end{thm}

\begin{proof}
Note that for $m=1$, such a problem instance is given by the conflict
in Example~\ref{ex:simple}. We obtain the general case by raising this
example to the ``$m$-th power'', as follows.

Our instance of {\sc Summarize} has $d = 3m$ dimensions, where every
dimension has a tree with root $r_i$ and two children $a_i$ and
$b_i$. We group the dimensions into sets of size three, and define for $i
\in \{ 1, \ldots, m \}$:
$$S_i := \left\{ (a_{3i-2}, b_{3i-1}, r_{3i}),
(b_{3i-2}, r_{3i-1}, a_{3i}), (r_{3i-2}, a_{3i-1}, b_{3i}) \right\}$$
Note that these are copies of Example~\ref{ex:simple} restricted to
dimensions $3i-2, 3i-1, 3i$.  The weight function $w$ is 1 for all nodes in
$S := S_1 \times S_2 \times \cdots \times S_m$; all other
nodes have weight zero. We set $k = \infty$.

By construction, $S$ has $3^m$ elements, and therefore a total weight
of $3^m$. For the first claim, we now show that $S$ is
overlap-free. Consider two different elements of $S$. Clearly, they
must differ on some factor $S_i$. But by definition of $S_i$ that
means that they are disjoint.

For the second claim, we now inductively prove that every
conflict-free solution of this problem instance has a weight of at most
$2^m$. For $m=1$, the claim follows from Example~\ref{ex:simple}. For
$m \geq 2$, let $T \subseteq S$ be a conflict-free solution. We can
assume $|T| > 1$, since otherwise $T$ clearly is of size less than
$2^m$.

Since $T$ is conflict-free, there must be a dimension $i$ such that
$x_i \in \{ a_i, b_i \}$ for all $x \in T$ (as per the definition of a
conflict). By symmetry of our construction, we can wlog assume $i =
d$. Then $T$ can be decomposed as a disjoint union $T = T_a \cup T_b$,
where $T_a = T \cap \{ x | x_d = a_d \}$, and $T_b = T \cap \{ x | x_d
= b_d \}$.  By construction of $S$, the nodes in $T_a$ have the same
values in factor $S_m$. Removing the last three dimensions from $T_a$,
we obtain a set of nodes in $S_1 \times \cdots \times S_{m-1}$ with
the same cardinality. This set is also conflict-free (otherwise $T_a$
would contain a conflict), and forms a solution to the instance of
size $m-1$.  By induction, $T_a$ can have size at most $2^{m-1}$.
Similarly, $|T_b|\leq 2^{m-1}$, and therefore $|T| = |T_a| + |T_b|
\leq 2^m$.
\end{proof}

\begin{remark} [Role of Conflicts]
  \label{rem:conflicts}
  Note the fundamental role played by conflicts in the proofs of Theorem~\ref{thm:ThreeTreeHardness} and in Theorem~\ref{thm:lb}. The simple conflict in Example~\ref{ex:simple} underlies the constructions in both proofs. As stated by Lemma~\ref{lem:optimal-conflict-free}, in the absence of conflicts, {\sc Summarize} can be solved optimally. Therefore an interesting open question is to ask how frequently large-weight, non-overlapping conflict structures arise in practice. In the context of summarizing metric changes, it is likely that these are fairly rare because for Example~\ref{ex:simple} to manifest, there have to be three fluctuations, each from three separate causes, but the causes are such that they don't overlap with each other. It would be worthwhile to test this conjecture in practice.
\end{remark}

\begin{remark} [Dense input versus sparse input]
There is a significant difference between our work and
Multi-structural Databases~\cite{Fagin} in how the input is provided. In
\cite{Fagin}, the input consists only of the subset $V'$ of nodes that
have non-zero weight. Let us call this a \emph{sparse} input. In
contrast, we assume that the weights are specified for every node $v
\in V$. That is, we assume that the input is \emph{dense}. We chose
this modeling assumption because in practice (i.e. for the
applications in Section~\ref{sec:application}), we found that almost
all nodes in $V$ had non-zero weight. Even though our algorithm is
described for the dense case, it is straightforward to apply it to the
sparse case as well.

Sparseness plays a critical role in the hardness results
of~\cite{Fagin}. They perform a reduction from independent set (a
well-known NP-hard problem). Their reduction can be done in polynomial
time only if the input is sparse. In fact, in their reduction, the
number of nodes is equal to the number of trees.

These results therefore do not imply that the dense case is also
NP-hard. In principle, it is possible that NP-hardness disappears
because we `only' need to be polynomial in the product of the sizes of
the trees. Theorem~\ref{thm:ThreeTreeHardness} shows that the problem
is NP-hard even with dense input. So the hardness is not due to
the density of input, but due to the presence of conflicts.
\end{remark}

\begin{remark} [Comparison to Rectangle Packing]
  It is instructive to compare {\sc Summarize} to the problem of the max-weight packing of axis-aligned hyper-rectangles in a space so that no two rectangles spatially overlap. Nodes in our problem correspond to hyper-rectangles in that setting (cf~\cite{Muthu,Berman,Agarwal}). That problem is optimally solvable for one dimension and NP-Hard for two more more dimensions. It can be shown that every instance of our problem is an instance of that problem, but not vice versa. This is because a tree cannot model the case where two hyper-rectangles intersect along a certain dimension, but neither contains the other. What Theorems~\ref{thm:ThreeTreeHardness} and~\ref{thm:twotree} together show is that the restriction to hierarchical hyper-rectangles now allows positive results for two dimensions and ``postpones'' the NP-Hardness to three or more dimensions.  We borrow some proof ideas for Theorem~\ref{thm:log} from~\cite{Agarwal}.
\end{remark}

\section{Applications of the algorithm}
\label{sec:application}

The Cascading Analysts algorithm is fairly general. The key choices
when applying the algorithm to a specific context are to pick the
metrics and dimensions to apply them over, and a sensible weight
function. We have applied the Cascading Analysts algorithm to helping
advertisers debug their advertising campaigns via a report called the
``top movers report'' \cite{TopMovers}, and to helping websites
analyze their traffic within Google Analytics
\cite{AutomatedInsights}.

\subsection{Interesting Weight Functions}
\label{sec:weight-function}

In Section~\ref{sec:problem}, we discussed a very simple weight function used to analyze metric changes. For each node $v$, $w(v)$ was set to the absolute value of the difference in the metric values for the node between two time periods $|t_v - l_v|$, where $l_v$ is the metric value for the current time period, and $t_v$ is the  metric value for a past time period (pre-period). In this section, we present other alternatives that result in different types of summarization.

\subsubsection{Modeling Composition Change}
\label{sec:composition}
If the data displays a generally increasing trend (or a generally decreasing trend), it is possible that almost all the slices are data are generally increasing. So the weight function $w(v) = |t_v - l_v|$ essentially becomes $w(v) = t_v - l_v$, and the root node is a degenerate, optimal solution, because it has at least as much weight as any set of non-overlapping nodes. In practice we may still want to separate low growth slices from high growth ones, because the former could still be improved. A simple option is to compare the mix or the composition of the metric rather than magnitude of the metric, that is
\begin{equation}
  \label{eq:composition}
w(v) := |\frac{t_v}{\sum_{v \in V} t_v}  - \frac{l_v}{\sum_{v \in V} l_v}|
\end{equation}

This way, the output of {\sc Summarize} consists of nodes that were a large fraction of the market in one of the two time periods, and a relatively small fraction in the other.

This technique is also useful in performing Benchmarking~\cite{Benchmarking}. In Benchmarking the goal is to compare the metric for a business against similar businesses. For instance, comparing the traffic of one messaging app against the traffic of another app across dimensions such as location, user age, etc. Here, $t_v$'s correspond to amount of traffic for the protagonist's company/app and $l_v$'s represent traffic for a benchmark app. It is usually the case that one of the businesses/apps is much larger, and therefore it makes more sense to compare the composition of their markets as in Equation~\ref{eq:composition}.

\subsubsection{Modeling Absolute Change v/s Relative Change}
\label{sec:box-cox}

A slightly different issue is that a large relative change in a metric (say from $\$500$ to $\$1000$) may be  more interesting than a small relative change in metric (say from $\$10\,500$ to $\$11\,000$),  because the latter is probably due to usual fluctuations (noise), whereas the former is a real event worth responding to. However, focusing entirely on relative change could produce tiny segments as output (say the metric changes from $1$ to $100$). In practice, it makes sense to weigh relative changes to some extent. To model this, we apply a standard technique called a Box-Cox transform~\cite{BoxCox} by setting the weight $w(v) := \frac{l_v^{1-m} - t_v^{1-m}} {1-m}$. For $m=0$, this reduces to the absolute value of the difference $|l_v - t_v|$. When $m \rightarrow 1$, this approaches $\log(l_v) - \log(t_v)$, which models a relative difference. In practice we found it useful to set $m$ in the range $[0.1, 0.3]$.

\subsection{Analyzing the Performance of Ad Campaigns}

Google's Adwords platform supports advertisers that run pay-per-click advertising campaigns. There are many factors that affect the performance of a campaign~\cite{AdWords}.
The advertiser's own actions, an advertiser's competitor's actions, a change in user behavior, seasonal effects, etc. Given that large amounts of money are often at stake, most advertisers monitor and react to changes in their spend very carefully.

A routine task for advertisers is therefore to (1) periodically identify the ``top moving'' data-segments, and (2) react to these changes sensibly. We seek to automate (1). Adapting the cascading analysts algorithm involves carefully choosing the dimensions and the weight function.





\begin{figure*}
  \centering
    \includegraphics[width = 4in] {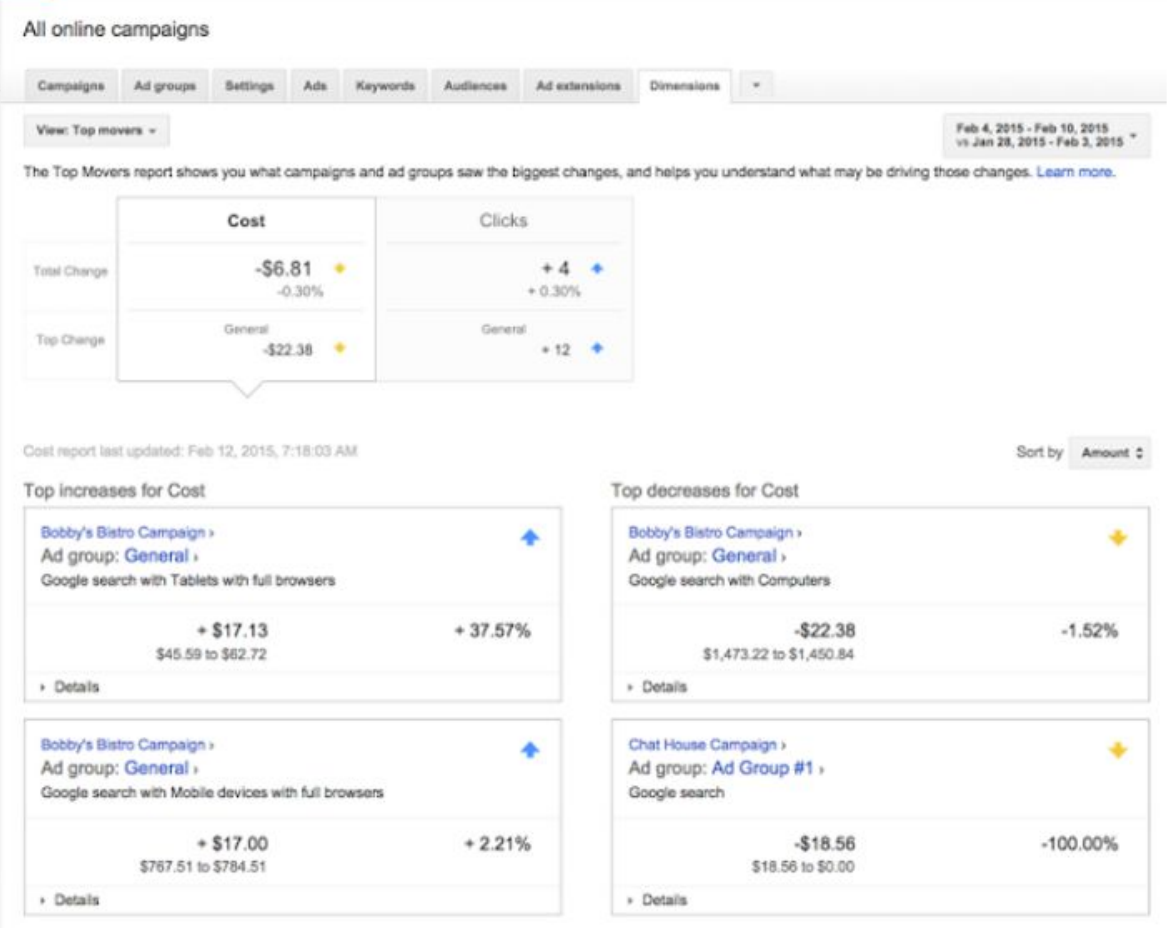}
  \caption{Google Adwords ``Top Movers'' report, showing traffic segments responsible for growth and decline of spending.}
\end{figure*}


We use three hierarchical dimensions to partition campaign performance. The first is the campaign-adgroup-keyword hierarchy.
Each campaign can have multiple adgroups, and an adgroup can have multiple keywords, forming a tree structure.
For example, a gift shop can have
one campaign for flowers and one campaign for cards. Within the flower campaign,
it can have one adgroup for each zip code that the shop serves.
A second hierarchical dimension partitions the user by the kind of device they use. The top level split can be Mobile, Desktop, and Tablet, and each of these can further be split by the make and model of the device.
And the third dimension is the platform on which the ads were shown, for instance, Google search, on the Web, or on Youtube.








We briefly describe how the weight function is modeled. The metrics of interest are the spend of the advertiser, the number of clicks received by the ads, and the number of views (impressions) received by the ads. We usually compare two time-slices. The weight function is modeled as the BoxCox transformation applied to the values of the metric in the two time-periods (see Section~\ref{sec:box-cox}).

\subsection{Understanding Website Traffic}

Google Analytics helps website and app developers understand their users and identify opportunities to improve their website (and similarly for phone apps)~\cite{GA}. There are many factors that affect the traffic to a website.
Changes in user interests, buzz about the website in social media, advertising campaigns that direct traffic to the website, changes to the website that make it more or less engaging, etc. The ``Automated Insights'' feature of Google Analytics~\cite{AutomatedInsights} analyzes the traffic data and identifies opportunities or room for improvement. The Cascading analysts algorithm is used within a section of this feature that identifies focus-worthy segments of traffic. The feature involves other algorithms and curation to make the insights actionable.

We now discuss the metrics and the dimensions. Google Analytics is used by a variety of apps and websites, for example by content providers like large newspapers, ecommerce platforms, personal websites or blogs, mobile apps, games, etc.
Different dimensions and metrics are important for different businesses. Consequently, Google Analytics has a very large set of dimensions and metrics that it supports. Some metrics include visits to the website, number of users, number of sessions, and a metric called goals whose semantics are user-defined. Some examples of dimensions include the source of the traffic to the website (search engines, social network sites, blogs, direct navigation), medium (was the visit due to an ad, an email, a referral), geographic dimensions (continent, country, city), device related dimensions (as in our AdWords example above) etc. The Cascading analysts is applied to several coherent groupings of dimensions and metrics. For instance, we may run the algorithm to compare the composition (see Section~\ref{sec:composition} of visits in one month versus another, with three dimensions like source, geography and device. (Here we compare compositions rather than the raw metric magnitudes because large seasonal effects could make all the data trend up or down.) This produces several candidate segments that are then turned into insights reports.

\section{Related Work}
\label{sec:related}


\subsection{OLAP/Drill-Down}

There is a large body of literature on OLAP~\cite{OLAP}. As discussed in the introduction, there is justified interest in automating data analysis for it. There is work on automating or helping the automation of drill-downs~\cite{Sarawagi:2001:UMA:767141.767148, Sarawagi00user-adaptiveexploration, Sarawagi98discovery-drivenexploration,DBLP:journals/pvldb/GebalyAGKS14,Mampaey:2011:TMI:2020408.2020499}. These attempts to summarize patterns in the data use information-theoretic approaches rather than explain the difference in aggregate metrics. Candan et al~\cite{DBLP:conf/edbt/CandanCQS09} propose an extension to OLAP drill-down that takes visualization real estate into account, by clustering attribute values. Again, this is not targeted to a specific task like explaining the change in aggregate metrics.

The database community has devoted a lot of attention to the problem of processing taxonomical data. For example \cite{Ben-Yitzhak:2008:BBF:1341531.1341539, springerlink:10.1007:s10618-007-0063-0, Qi:2008:SOO:1376616.1376703, springerlink:10.1007:978-3-642-16373-9:27} consider the same or related models. Broadly, they concentrate on the design issues such as query languages, whereas we focus on computational issues and optimization.

There is recent work by Joglekar, Garcia-Molina, and Parameswaran~\cite{Smartdrill}, which we call Smart Drill-Down, that like us attempts to ``cover'' the data. They trade off dual objectives of ``covering'' the data and ``interestingness'' with guidance from an analyst. This trade-off is inherently data-dependent, because the two objectives are like apples and oranges. The analyst must try some trade-off and rebalance it based on the results. In contrast, as we discuss in Section~\ref{sec:application}, we have a single objective function that directly  models the summarization problem, and therefore our approximation bounds have clear meaning to an analyst or decision-maker. So while the formulation in Smart Drill-Down~\cite{Smartdrill} may be more general, our solution is more automated for our specific problem. Two other differences between their work and ours is that they allow double-counting, which as we discussed is undesirable for our application, and their algorithm is top-down. Indeed, most of OLAP is inherently about top-down exploration. But a top-down algorithm may omit important events in the data. For instance, if you have two adjacent zipcodes in a metro, one with an increase in unemployment, and another with a decrease, the two phenomena could cancel in a way that the net effect on the metro is nearly zero. OLAP on the metro would seem to suggest that nothing interesting is happening within the metro, whereas there may be.

There is some work~\cite{Bu:2005:MSH:1083592.1083644, Lakshmanan:2002:GMA:1287369.1287435, DBLP:conf:kdd:XiangJFD08, Geerts04tilingdatabases} on finding hyper-rectangle based covers for tables. In contrast, our work is about ``packing'', i.e., we explicitly deal with the double-counting issue.

\subsection{Multistructural Databases}
We are directly inspired by the problem formulations of Fagin et al.~\cite{Fagin} and Sarawagi~\cite{Sarawagi}.
Fagin et al~\cite{Fagin} formulate three operators as optimization problems over lattices---{\sc Divide} finds balanced partitions of data, {\sc Differentiate} compares two slices of data along some dimensions to find regions that account for the difference (this is the scenario of Example~\ref{ex:running} in the introduction), and {\sc Discover} finds cohesive, well-separated clusters of data. (Sellam and Kersten~\cite{sellam:meet} also work a formulation similar to {\sc Divide} in the context of designing a query language for exploration.). {\sc Discover} and {\sc Differentiate} are algorithmically similar, differing only in how the weights on the hypernodes are computed. Our algorithms apply to both of these operations. They show that these operators are hard to approximate on hierarchies represented as lattices, and present dynamic programming algorithms that provide approximations for a sequential composition of tree-like or numeric dimensions. Fagin et al.~\cite{Fagin2} extends these results to a wider class of operators implementable via objective functions that are distributive on the slices, discuss different functions to combine the weights of the output segments, and presents experimental validation for their techniques on a large data set. In a series of papers, Sarawagi et al~\cite{Sarawagi, Sarawagi2, Sarawagi3} discuss information-theoretic formulations that help automate the exploration of datacubes that store aggregates---in the sense of our model, the focus is on designing the objective function (including the weight function). The algorithmic results are similar to the positive results of Fagin et al\cite{Fagin} discussed above.

Our main contribution is to identify practically relevant restrictions of these models, and to supply interesting algorithms. The hardness results in both papers were devastating because the models were overly general (they used lattices instead of a product of a small number of trees). In contrast, their algorithmic results were restrictive. Recognizing that the problem is hard to solve over a lattice, both papers extrinsically convert the lattice into a tree by \emph{sequentially} composing dimensions (pasting one dimension below another). This precludes certain data-cubes from being candidate solutions.

\begin{example} \label{mutex}
Consider US employment data in a two dimensional space: Location $\times$ Gender. If we compose the dimensions strictly as location before gender, then the following pair of non-overlapping nodes will never be considered in the same solution: (North-east, Male) and (New York, Female). We can only reach both of these nodes by first splitting the space on gender, and then splitting each of subspaces independently on location.
\end{example}

One may try all possible sequential compositions of the dimension trees - arguably an efficient operation when the number of dimensions is a constant - and pick the best one. The following example proves that even with this improvement, we can only expect a $\Omega(n^{1/4})$-approximate solution. By Theorem~\ref{thm:twotree}, our algorithm is optimal for two dimensions. So, we have a significant improvement even for two dimensions.

\begin{example} \label{mutex-bad}
Consider an instance of \problemname with two dimensions, each with the same topology: a star with $\sqrt{m}$ strands (paths), each with $\sqrt{m}$ nodes, and a root. 
Suppose that the strands are labeled $1 \ldots \sqrt{m}$ from left to right, and the \TVs within a strand are labeled $1 \ldots \sqrt{m}$, then we can label each \TV by a pair (strand index, within-strand index). Notice that $n = (m + 1)^2$.

Suppose further that the weights are in $\{0,1\}$. The only lattice nodes with a weight of $1$ are the $m$ nodes indexed by the quadruple $(i,j),(j,i)$ for $i \in \{1 \ldots \sqrt{m}\}$ and  $j \in \{1 \ldots  \sqrt{m}\}$. The optimal solution has value $m$ because none of the non-zero weight nodes overlap---if two nodes overlap in one dimension, then they belong to different strands in the other dimension.

Now, by Observation~\ref{mutex}, for a sequential composition $T_1$ followed by $T_2$, we cannot pick a pair of lattice nodes that belong to the same strand in dimension $T_1$, yielding a solution of size at most $\sqrt{m}$. The argument for the other sequential composition is symmetric, and we have at best a $\sqrt{m}$ approximation.
\end{example}

\section{Conclusion}
We study the problem of summarizing metric changes across hierarchical dimensions. Our main contribution is a new algorithm called Cascading Analysts that is a significant improvement over prior algorithms for this problem. We apply this algorithm to two widely used business-intelligence tools (Google Adwords and Google Analytics) that help advertisers and website/app developers analyze their core metrics.

Studying concrete applications gave us an interesting lens on the computational hardness of the summarization problem. We identified a practically relevant restriction of the previously studied model of hierarchical data to a product space of trees -- prior work studied lattices. Without this restriction, the problem on lattices is computationally hard, although, this hardness is relatively uninteresting, i.e., it stems purely from the ``dimensionality'' of data. In practice, we note that the summarization problem is useful to solve even when there are a few (say less than five) dimensions as in our applications. Further investigation reveals a more interesting source of hardness---the presence of structures called \emph{conflicts} that occur in data with three or more dimensions. Fortunately, this source of hardness does not preclude approximation.

One direction of future work is to better understand the prevalence of \emph{conflicts} in an ``average case'' sense. Our belief (see Remark~\ref{rem:conflicts}) is that large weight conflicts ought to be rare in practice. It would be interesting to formalize this in a beyond-worst-case-analysis sense~\cite{beyond}.

\bibliographystyle{abbrv}
\bibliography{drill}

\end{document}